\documentclass[journal]{IEEEtran}

\ifCLASSINFOpdf
\else
   \usepackage[dvips]{graphicx}
\fi
\usepackage{url}
\usepackage{caption}
\usepackage{subcaption}
\usepackage{amssymb}
\usepackage{amsmath}
\usepackage{amsthm}
\usepackage{mathrsfs}
\usepackage[nolist,nohyperlinks]{acronym}
\usepackage{xcolor}
\usepackage{soul}
\usepackage{nameref}

\newtheorem{theorem}{Proposition}
\newtheorem{definition}{Definition}
\usepackage{graphicx}

\usepackage[style=ieee]{biblatex}
\begin{acronym}
\acro{KT}{kaleidoscope transform}
\acro{MRI}{magnetic resonance imaging}
\acro{DFT}{discrete Fourier transform}
\acro{CS}{compressed sensing}
\acro{ChaoS}{chaotic sensing}
\acro{CT}{computed tomography}
\acro{PET}{positron emission tomography}
\acro{RMSE}{root-mean-squared-error}
\acro{PSNR}{peak signal-to-noise ratio}
\acro{SSIM}{structural similarity}
\acro{VIF}{visual information fidelity}
\acro{SVD}{singular value decomposition}
\end{acronym}

\newcommand{\floor}[1]{\left\lfloor #1 \right\rfloor}
\begin{document}

\title{Bespoke Fractal Sampling Patterns for Discrete Fourier Space via the Kaleidoscope Transform}

\author{Jacob M. White, \IEEEmembership{Member, IEEE}, Stuart Crozier, and Shekhar S. Chandra, \IEEEmembership{Member, IEEE}

\thanks{J. M. White, S.S. Chandra and S. Crozier are with the University of Queensland, Brisbane, St Lucia QLD 4072, Australia (e-mail: uqjwhi35@uq.edu.au)}}
\markboth{Journal of \LaTeX\ Class Files, Vol. 14, No. 8, August 2015}
{Shell \MakeLowercase{\textit{et al.}}: Bare Demo of IEEEtran.cls for IEEE Journals}
\maketitle

\begin{abstract}
Sampling strategies are important for sparse imaging methodologies, especially those employing the \ac{DFT}. \Acl{ChaoS} is one such methodology that employs deterministic, fractal sampling in conjunction with finite, iterative reconstruction schemes to form an image from limited samples. Using a sampling pattern constructed entirely from periodic lines in \ac{DFT} space, \acl{ChaoS} was found to outperform traditional compressed sensing for \acl{MRI}; however, only one such sampling pattern was presented and the reason for its fractal nature was not proven. Through the introduction of a novel image transform known as the \acl{KT}, which formalises and extends upon the concept of downsampling and concatenating an image with itself, this paper: (1) demonstrates a fundamental relationship between multiplication in modular arithmetic and downsampling; (2) provides a rigorous mathematical explanation for the fractal nature of the sampling pattern in the \ac{DFT}; and (3) leverages this understanding to develop a collection of novel fractal sampling patterns for the 2D \ac{DFT} with customisable properties. The ability to design tailor-made fractal sampling patterns expands the utility of the \ac{DFT} in chaotic imaging and may form the basis for a bespoke \acl{ChaoS} methodology, in which the fractal sampling matches the imaging task for improved reconstruction. 

\end{abstract}


\IEEEpeerreviewmaketitle

\section{Introduction}


\IEEEPARstart{T}{he} discrete Fourier transform (DFT) \acused{DFT}is an important and fundamental transform in many signal and imaging areas. One such area is \ac{MRI}, which visualises soft tissue in high resolution without the use of ionising radiation \cite{Ehman2017}; however, with long acquisition times and commensurately greater costs than other imaging modalities \cite{Hollingsworth2015} it often goes unused.

Arguably the most impactful development towards addressing the long acquisition times of \ac{MRI} is that of \ac{CS} \cite{Candes2006, Donoho2006} which, in the case of \ac{MRI}, allows an image to be significantly undersampled using a random sampling pattern in the \ac{DFT} of image space (often referred to as k-space) and then reconstructed with an iterative algorithm \cite{Lustig2007}; however, as in many practical applications, truly random measurements are infeasible in \ac{MRI}. Fortunately, alternatives exist. Yu \textit{et al.} \cite{Yu2010} developed a hardware-friendly sampling scheme which approximates the desired properties of random sampling by using sensing matrices populated with chaotic sequences. Theirs and subsequent works have demonstrated reconstruction performance on par with or better than random sampling \cite{Kafedziski2011, Zeng2014, Rontani2016, Chandra2018, Gan2019, Gan2019a}.

Similarly, \ac{ChaoS} proposes the use of a deterministic fractal sampling pattern for the \ac{DFT} and \ac{MRI}~\cite{Chandra2018}. The noise introduced by this sampling scheme is turbulent and thus image-independent, and may be easily removed through dampening and obtaining the maximum likelihood solution. Furthermore, as the fractal is constructed entirely from periodic lines in the \ac{DFT} based on the discrete Radon transform~\cite{Matus1993}, it may be feasibly acquired using MRI hardware \cite{Chandra2018}. \ac{ChaoS} appears to be a promising alternative to \ac{CS} for MRI, having been shown to outperform radial sampling and compressed sensing in \ac{SSIM} and \ac{VIF} when applied both to the Shepp-Logan phantom image with Gaussian noise added and to data experimentally gathered using a phantom constructed from Lego blocks in a liquid-filled tube \cite{Chandra2018}; however, although it was verified numerically that the \ac{ChaoS} pattern is a fractal, it is not well understood why this fractal structure arises \cite{Chandra2018}. 

In this work, a novel image transform known as the \ac{KT} is presented, which formalises and extends upon the concept of downsampling and concatenating an image with itself. Through proving that certain scaling operations in the \ac{DFT} are equivalent to \acp{KT}, the self-similar, fractal nature of the \ac{ChaoS} fractal is explained. This understanding is then leveraged to develop novel, customisable fractal sampling patterns for use in \ac{ChaoS} or any other methodology designed to employ the mixing of chaotic imaging information involving the \ac{DFT}. 


This paper is structured as follows. In the remainder of this section, the method of construction of the \ac{ChaoS} fractal is provided, and the potential utility of customisable variations of this pattern is justified. Section II then introduces the \ac{KT} and proves its relationship to scaling in the \ac{DFT}. Finally,  section III utilises the theory of the \ac{KT} to explain the fractal nature of the \ac{ChaoS} sampling pattern and presents novel fractal sampling patterns along with methods for their construction. 

\subsection{Chaotic Sensing Fractal}
The method for generating the fractal sampling pattern used in \ac{ChaoS} on an $N \times N$ grid, where $N$ is prime, as described in \cite{Chandra2018}, begins with generating the Farey sequence of order $N$, $F_N$. This is the sequence of irreducible fractions in the domain $[0, 1]$ with denominator less than or equal to $N$. This sequence may be easily generated using the mediant property of the Farey sequence. Beginning with $\frac{a_1}{b_1} = \frac{0}{1}$ and $\frac{a_2}{b_2} = \frac{1}{1}$, the mediant fraction, $\frac{a_3}{b_3}$, may then be generated using the rule
\begin{align*}
\frac{a_3}{b_3} = \left( \frac{a_1 + a_2}{b_1 + b_2} \right).
\end{align*}
This may be performed recursively, stopping when $\frac{a_3}{b_3} = \frac{1}{N}$, to generate all elements of the sequence~\cite{Hardy1979}.

Each of these fractions is then mapped to a grid position using the relation $\frac{a}{b} \to (b, a)$, where $(b, a)$ represents $b$ pixels across and $a$ pixels up. It may be observed that these points represent all possible directions from the origin in the first octant of the grid. To generate the points corresponding to the rest of the plane, this octant is flipped and mirrored. 

The generated points are then sorted by their Euclidean distance from the origin ($L^2$ norm), $\ell_2 = a^2 + b^2$, and those closest to the origin are selected (the exact number chosen is dependent on the Katz criterion, which will be defined shortly). Each point is then mapped to a periodic line given by all of its integer multiples computed modulo $N$, \textit{i.e.} any points on the line that would lie outside the grid are wrapped around back into it. 

As mentioned, the number of periodic lines required is governed by the Katz criterion. This is given by
\begin{equation*} 
K = \frac { \max \left ({\sum _{j=0}^{\mathcal {N}-1} |a_{j}|, \sum _{j=0}^{\mathcal {N}-1} | b_{j}|}\right)}{N},
\end{equation*}
where the minimum information required for an exact reconstruction is when $K = 1$ \cite{Katz1979}; however, $K$ may be greater than this. Using the box-counting method, it was determined that the fractal constructed with $N = 257$ had a Minkowski-Bouligand dimension of 1.79; however, no attempt was made to find this dimension analytically and the reason why the self-similar, repeated copies of the central circular pattern arose from the described construction could not be explained.


\subsection{Fourier Signatures}
Work has been conducted in the past on using the unique \ac{DFT} characteristics of different image classes, known as Fourier signatures, to tailor-make MRI sampling patterns for specific tasks. Two such examples are feature-recognising \ac{MRI} \cite{Cao1993} and MRI with encoding by \ac{SVD} \cite{Zientara1994}, which were found to be equivalent given the same training data \cite{Cao1995}.

Using feature-recognising \ac{MRI}, $32 \times 16$ pixel simulated images were able to be significantly better reconstructed than when using a direct inverse Fourier transform; however, the authors did not apply this technique to real \ac{MRI} data. Furthermore, no consideration was given to the nature of the sampling determined by this method. It is foreseeable that for a given family of images it would be just as difficult to implement as the random pattern suggested by \ac{CS}. Thus, although these two techniques demonstrate that there can be advantages to using sampling patterns tailored to particular image classes in \ac{MRI}, neither is able to guarantee that these patterns are physically implementable on \ac{MRI} hardware or that they lead to any significant acquisition acceleration.

The \ac{ChaoS} fractal, on the other hand, is theoretically guaranteed to be feasibly implementable in hardware due to its construction from periodic lines in the \ac{DFT}. Thus, a method to generate custom \ac{ChaoS} fractals would combine the ability to tailor a sampling pattern to a specific image class, as in feature-recognising \ac{MRI}, with turbulent noise of \ac{ChaoS}.

\section{The Kaleidoscope Transform}
Given a positive integer downsampling factor, $\nu$, and integer smear factor, $\sigma$, the $\nu, \sigma$-\ac{KT} converts a sequence, image, or array of any dimension into $\nu$ evenly spaced, downsampled copies of itself along each axis, each scaled by a factor of $\sigma$, as presented in Fig. \ref{fig: Darter Transforms} and \ref{fig: Smear factor example}.

\begin{figure}[ht]
\centering
\includegraphics[width=\linewidth]{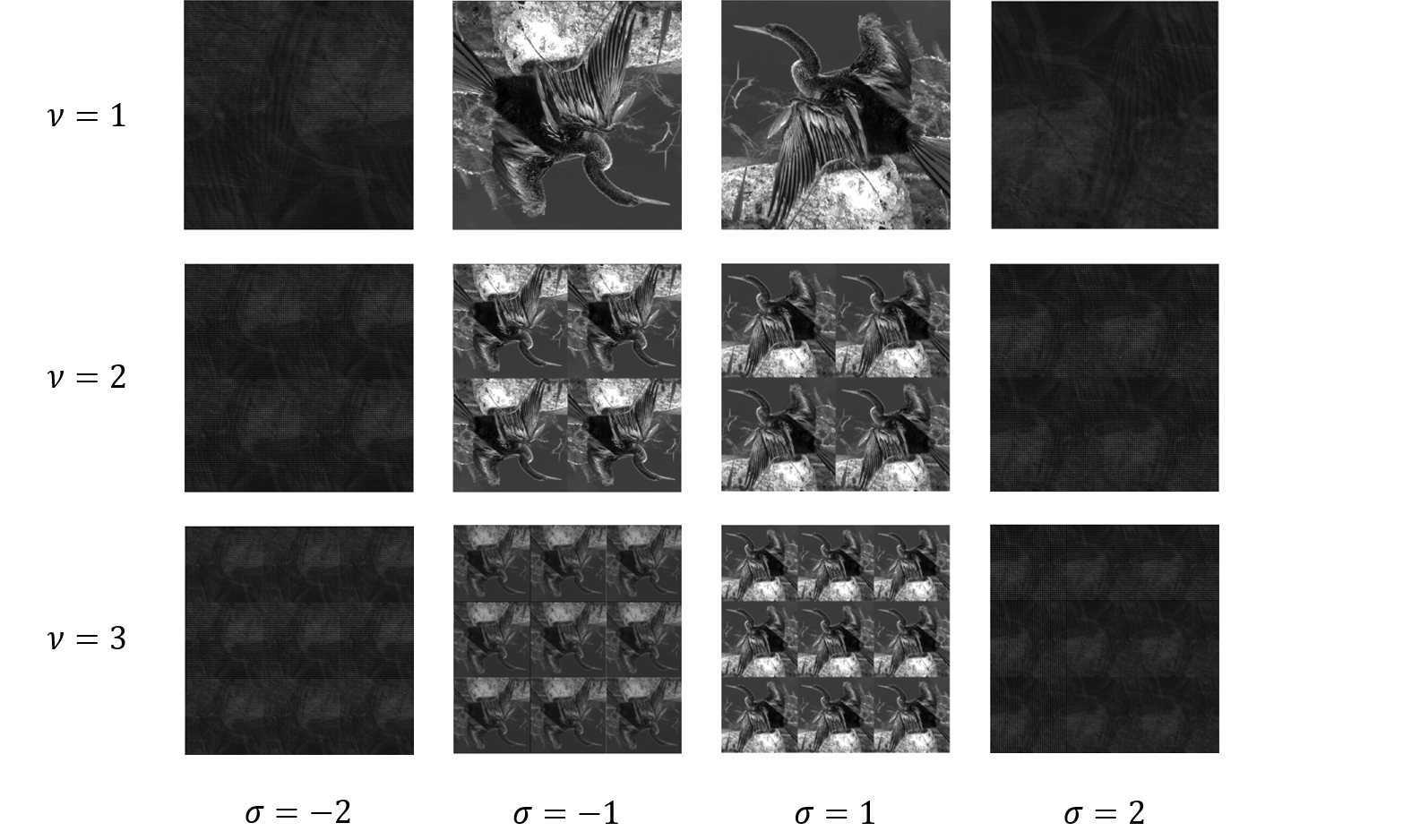}
\caption{$\nu,\sigma$-\ac{KT}s of a sample $512 \times 512$ image. The original image (source: \cite{UGR2003}) is presented in the $(\nu, \sigma) = (1,1)$ position.}
\label{fig: Darter Transforms}
\end{figure}

The \ac{KT} will be based on the kaleidoscope mapping and it is convenient to define two such mappings: an upper and a lower. The intuition behind these mappings is further explained in Supplementary Material A: \nameref{Supp A}.
\begin{definition}[Upper Kaleidoscope Mapping]
The upper $\nu, \sigma$-kaleidoscope mapping over $\mathbb{Z}_N$, where $N, \nu \in \mathbb{Z}^+$ and $\sigma \in \mathbb{Z}$, written $\kappa_U: \mathbb{Z}_N \to \mathbb{Z}_N$, is defined by \label{def: Upper kaleidoscope Mapping}
\begin{align*}
\kappa_U[n; \nu, \sigma] &=  \left( \left \lceil \frac{N}{\nu} \right \rceil (n \bmod \nu)  + \sigma \left \lfloor \frac{n}{\nu} \right \rfloor \right) \bmod N.
\end{align*} 
\end{definition}
\begin{definition}[Lower Kaleidoscope Mapping]
The lower $\nu, \sigma$-kaleidoscope mapping over $\mathbb{Z}_N$, where $N, \nu \in \mathbb{Z}^+$ and $\sigma \in \mathbb{Z}$, written $\kappa_L: \mathbb{Z}_N \to \mathbb{Z}_N$, is defined by \label{def: Lower kaleidoscope Mapping}
\begin{align*}
\kappa_L[n; \nu, \sigma] &=  \left( \left \lfloor \frac{N}{\nu} \right \rfloor (n \bmod \nu)  + \sigma \left \lfloor \frac{n}{\nu} \right \rfloor \right) \bmod N.
\end{align*} 
\end{definition}
These may then be combined into a single, unified kaleidoscope mapping as follows, where the notation of $0^+$ versus $0^-$ is used to distinguish between cases when $\sigma = 0$. 
\begin{definition}[Kaleidoscope Mapping]
The $\nu, \sigma$-kaleidoscope mapping over $\mathbb{Z}_N$, where $N, \nu \in \mathbb{Z}^+$ and $\sigma \in \mathbb{Z}$, written $\kappa: \mathbb{Z}_N \to \mathbb{Z}_N$, is defined by
\begin{align*}
\kappa[n; \nu, \sigma] &=  \begin{cases}
\kappa_L[n; \nu, \sigma], & \sigma \leq 0^{-} \\
\kappa_U[n; \nu, \sigma], & \sigma \geq 0^{+}.
\end{cases}
\end{align*} 
\end{definition}

At first the use of this notation may appear to restrict lower and upper kaleidoscope mappings to non-positive and non-negative smear factors respectively; however, it follows simply from the definition of each mapping that smear factors that are congruent modulo $N$ give equivalent mappings. Thus, any desired smearing may be achieved with either the upper or lower mapping.
\begin{figure}[ht]
\centering
\includegraphics[width = \linewidth]{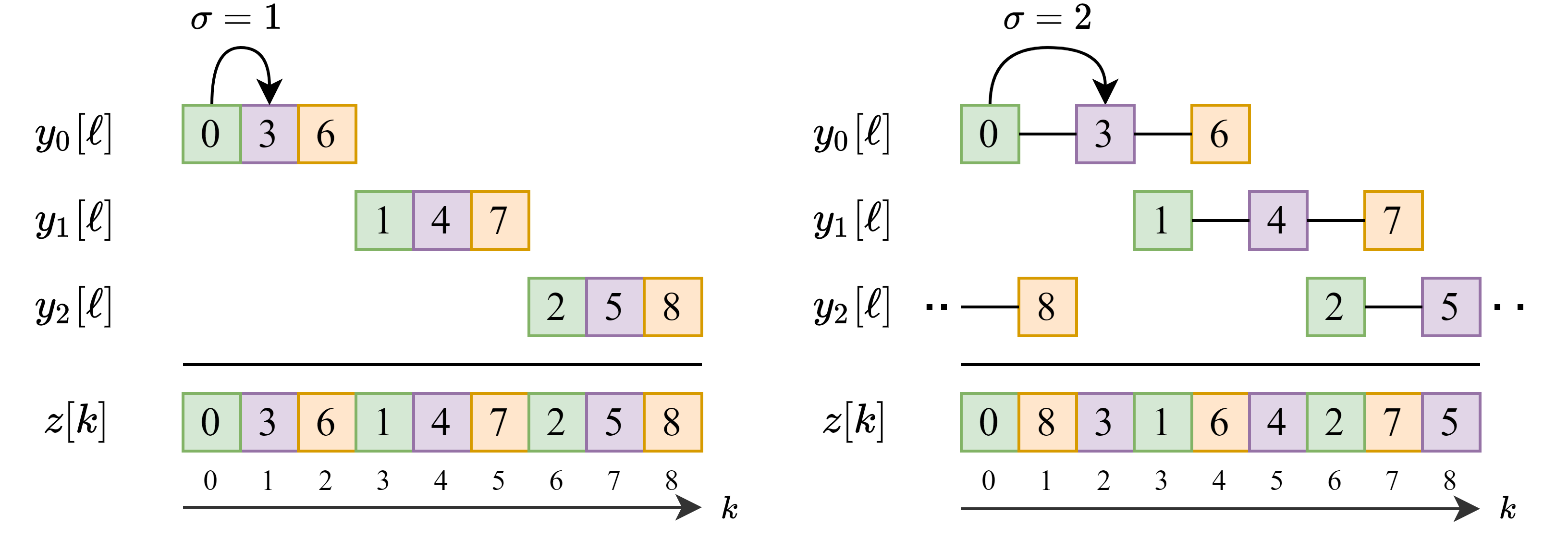}
\caption{The effect of altering smear factor, $\sigma$, on the downsampling and concatenating of the sequence $x[n] = n$ for $n \in \mathbb{Z}_9$, with a downsampling factor of $\nu = 3$. The cases with $\sigma = 1$ and $\sigma = 2$ are presented on the left and right respectively.}
\label{fig: Smear factor example}
\end{figure}

The \ac{KT} moves each element of a sequence to its new position determined by the kaleidoscope mapping, summing any elements that are mapped to the same position. It is defined as follows. 
\begin{definition}[Kaleidoscope Transform]
The  $\nu, \sigma$-\ac{KT} over $\mathbb{Z}_N$, where $N, \nu \in \mathbb{Z}^+$ and $\sigma \in \mathbb{Z}$, of a sequence $x[n]$ for $n \in \mathbb{Z}_N$, is defined by
\begin{align*}
\mathcal{K}_{\nu, \sigma}\{x[n] \}[k] &= \sum_{n=0}^{N-1} x[n] \cdot \delta[k - \kappa[n; \nu, \sigma]]
\end{align*}
for $k \in \mathbb{Z}_N$, where $\delta$ is the unit sample function and $\kappa$ is the kaleidoscope mapping over $\mathbb{Z}_N$. \label{def: kaleidoscope Transform}
\end{definition}
This definition may easily be generalised to higher dimensions by simply performing the kaleidoscope mapping (potentially with different $\nu$ and $\sigma$ values) on the index along each dimension.

We now prove an important result that relates the kaleidoscope mapping (and hence transform) to multiplication in modular arithmetic. 

\begin{theorem}[Kaleidoscope-Multiplication Theorem]
If given $N, \nu \in \mathbb{Z}^+$ and $\sigma \in \mathbb{Z}$ such that $|\sigma| < \nu$, we may write that $N = L\nu - \sigma$ for some $L \in \mathbb{Z^+}$, then
\begin{align*}
\kappa[n; \nu, \sigma] &= Ln \bmod N,
\end{align*}
where $\kappa[n; \nu, \sigma]$ is the $\nu,\sigma$-kaleidoscope mapping over $\mathbb{Z}_N$.
\end{theorem}
\begin{proof}
Suppose that the premise above holds. Now, we may observe from the definition of the kaleidoscope mapping that,
\begin{align*}
    \kappa[n; \nu, \sigma] &= \left( L (n \bmod \nu) + \sigma \left \lfloor \frac{n}{\nu} \right \rfloor \right) \bmod N.
\end{align*}
Now, by the definition of the mod operator,
\begin{align*}
\kappa[n; \nu, \sigma] &= \left( L \left( n - \nu \left \lfloor \frac{n}{\nu} \right \rfloor \right) + \sigma \left \lfloor \frac{n}{\nu} \right \rfloor \right) \bmod N\\
&= \left( L n - \left(  L \nu - \sigma \right) \left \lfloor \frac{n}{\nu} \right \rfloor \right) \bmod N \\
&= \left( L n - N \left \lfloor \frac{n}{\nu} \right \rfloor \right) \bmod N \tag{as $N = L\nu - \sigma$} \\
&= Ln \bmod N.
\end{align*}
\end{proof}

An important corollary of this result is that certain multiplications by an integer $L$ mod $N$ are equivalent to kaleidoscope mappings. Of particular interest are multiplications that result in unity-smear kaleidoscope mappings (those with $|\sigma|$ = 1). These occur when $L$ is a factor of $N \pm 1$. 

Multiplications by factors of $mN \pm 1$ for positive integers $m$ also approximate unity-smear kaleidoscope mappings. This is because such a multiplication mod $N$ may be expressed as the same multiplication mod $mN$, which is a unity-smear kaleidoscope mapping of the sequence zero-padded by a factor of $m$, followed by a reduction mod $N$. This reduction causes each repeated pattern to fill the zero-padded gaps of its neighbours, approximating a unity-smear kaleidoscope mapping. 


\section{Chaotic Sensing Through the Kaleidoscope}
The \ac{ChaoS} fractal on an $N \times N$ grid is traditionally thought of as the superposition of $R$ discrete periodic lines (where $R$ is determined by the Katz criterion), each given by $N$ multiples of an initial Farey vector; however, it may just as easily be expressed as the superposition of $\lceil N/2 \rceil$ scaled images, each containing the same multiple of all $R$ Farey vectors, along with their reflections, as presented in Fig. \ref{fig: CSTK}. It follows simply that those scaled images containing Farey vectors multiplied by factors of $mN \pm 1$ for positive integer $m$ will approximate unity-smear \ac{KT}s, containing downsampled, repeated copies of the original image (in the case when the vectors are sorted by their $L^2$ norm as in the classic fractal, a circle). Their superposition then gives the repeated, self-similar patterns observed in the \ac{ChaoS} fractal.

\begin{figure}
\centerline{\includegraphics[width=0.9\columnwidth]{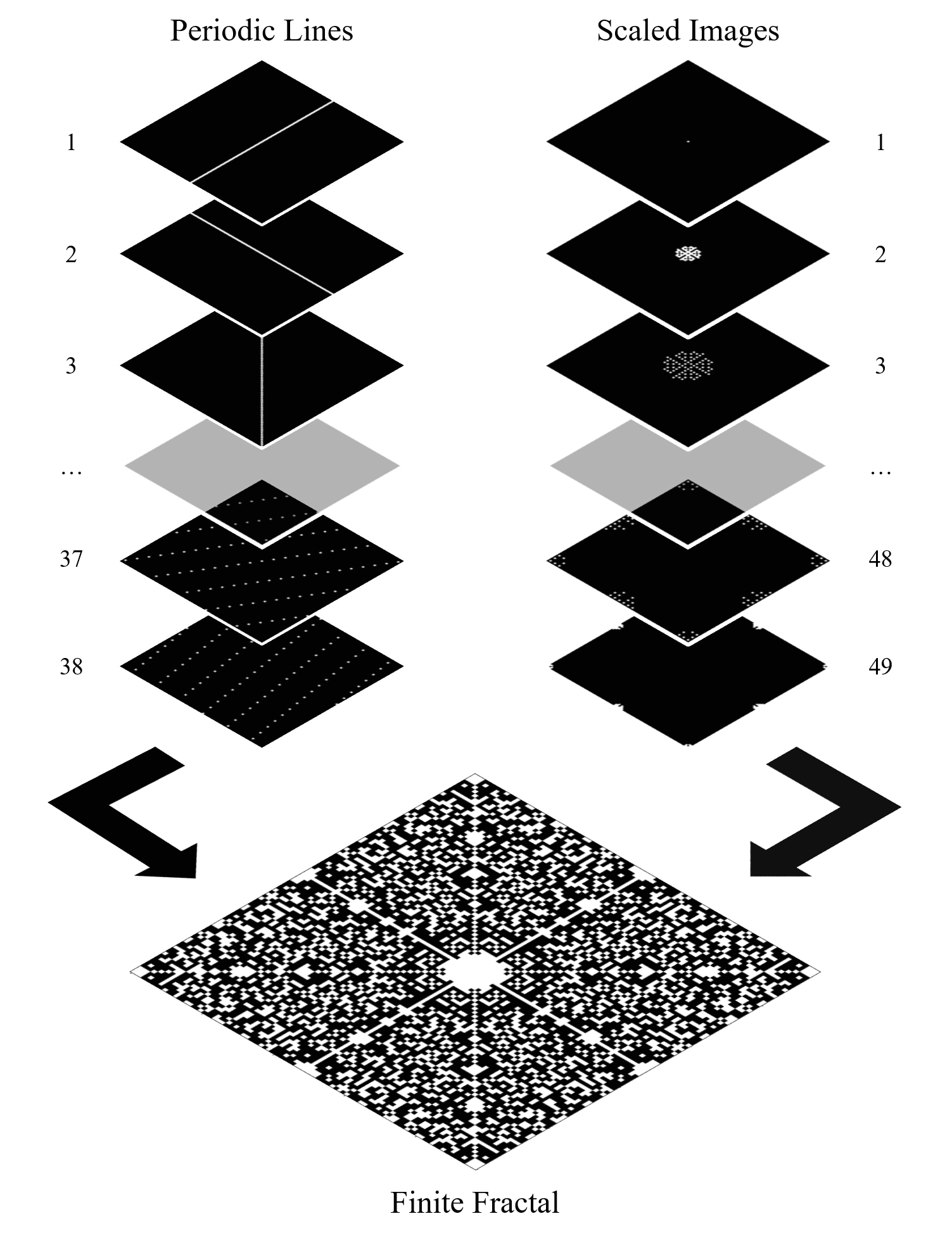}}
\caption{The ChaoS fractal on a $97 \times 97$ grid, represented as both the superposition of $R=38$ periodic lines and of $\lceil N / 2 \rceil = 49$ scaled images.}
\label{fig: CSTK}
\end{figure}

This theory for the fractal nature of the ChaoS fractal predicts that the repeated pattern need not be circular, and the grid on which the fractal is constructed need not be square (only that $M\pm1$ and $N\pm1$ share common factors for an $M \times N$ grid). Moreover, it predicts that fractals such as these may be generated in any number of dimensions. 

By using other $L^p$ norms to sort the Farey vectors used to generate the fractal, any superellipse may be realised as central pattern. Applying a linear transformation to each vector prior to calculating its norm allows for stretched and rotated variations of these central patterns. Furthermore, any star-shaped polygon may be realised as a central pattern by sorting the Farey vectors using what we refer to as the polygon quasi-norm. This calculates the minimum factor by which the desired polygon must be scaled for it to contain the given vector, as described in Fig. \ref{fig: n-quasinorm}. Note that in practice polygon central patterns are limited to those that are 180$^\circ$-rotationally-symmetric as the fractal is constructed from periodic lines. 

\begin{figure}[ht]
\centering
\includegraphics[width=0.9\linewidth]{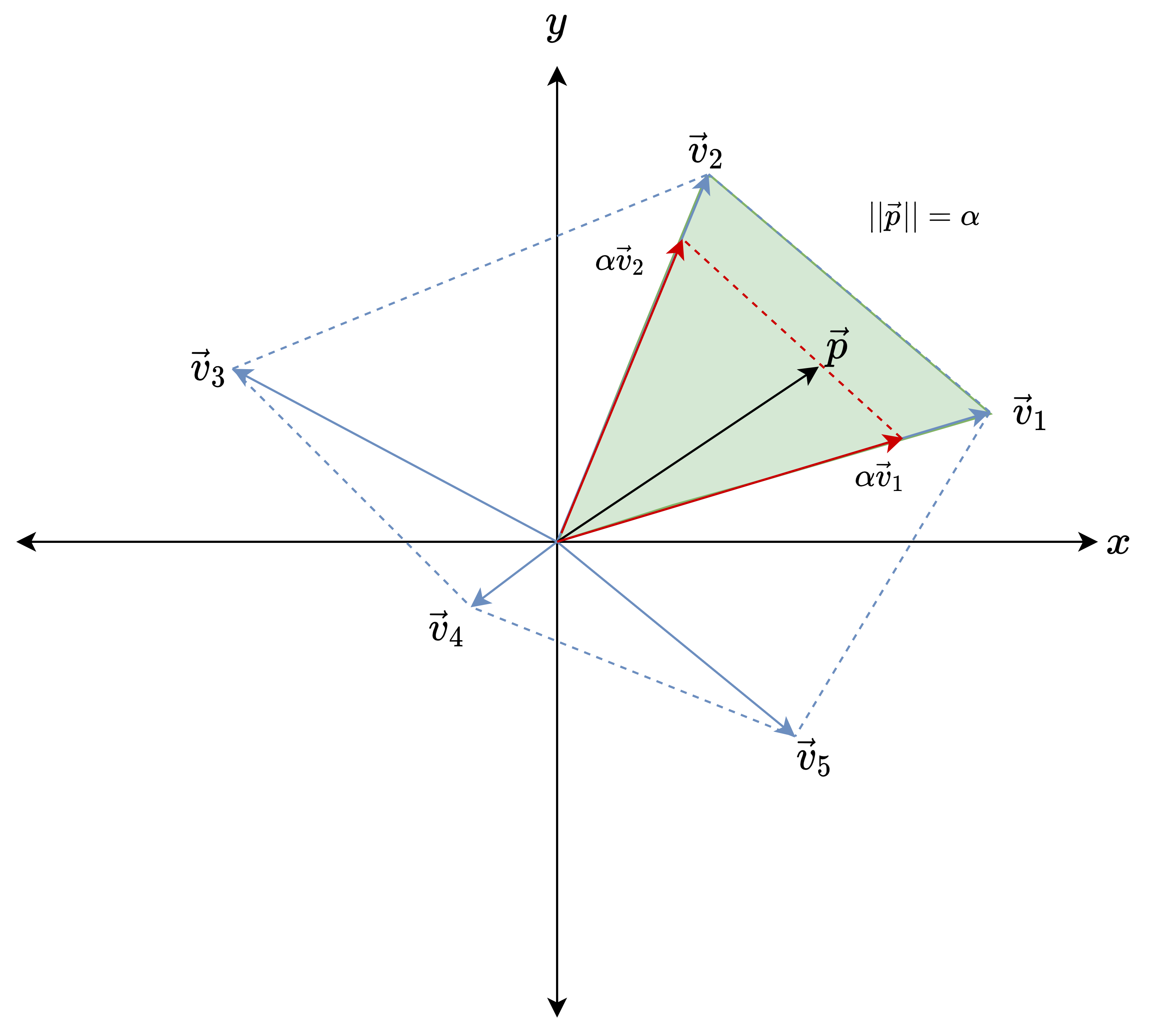}
\caption{The polygon quasi-norm of a vector $\vec{p}$ is computed by first determining the triangular region of the star-shaped polygon in which $\vec{p}$ lies (green region), and then calculating the scale factor $\alpha$ required such that $\vec{p}$ is co-linear with its neighbouring vertices.}
\label{fig: n-quasinorm}
\end{figure}

Further extensions to the \ac{ChaoS} fractal may be realised by diverging entirely from the paradigm of using periodic lines to construct the sampling pattern. Instead, the fractal may be explicitly constructed from the superposition of scaled copies of a desired pattern. Of particular interest are fractals constructed from feasibly implementable k-space trajectories other than lines, such as spirals. When using this method, we may also explicitly limit the scaled copies of the central pattern to those that give unity-smear \acp{KT}. Fig. \ref{fig: New fractals} presents a selection of such patterns for \ac{ChaoS} constructed using these methods. 

\begin{figure}[ht]
\centering
\begin{subfigure}{0.5\columnwidth}
  \centering
  \includegraphics[width=0.95\linewidth]{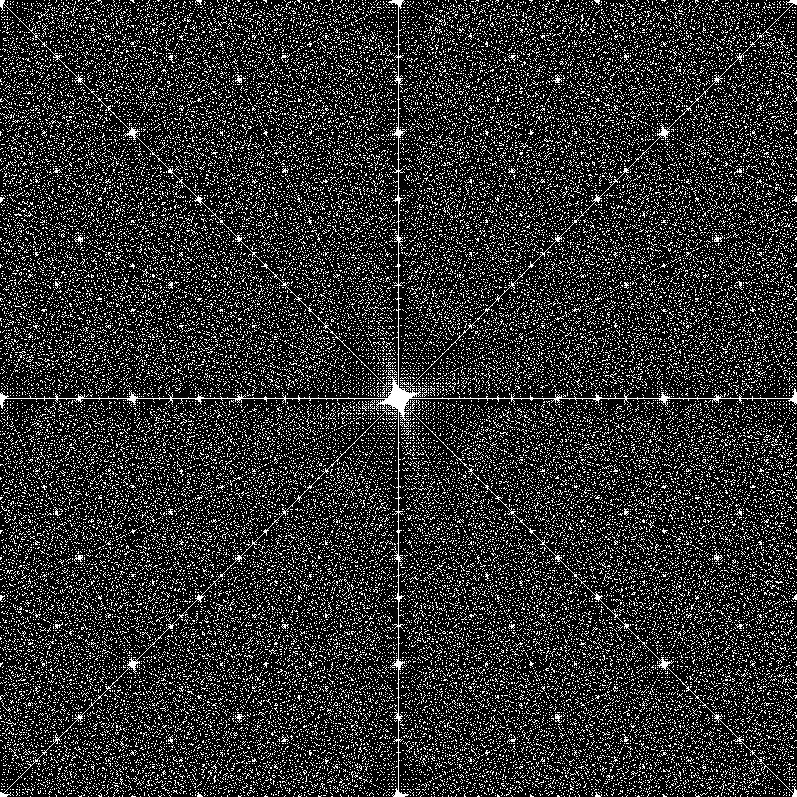}
  \caption{}
\end{subfigure}%
\begin{subfigure}{0.5\columnwidth}
  \centering
  \includegraphics[width=0.95\linewidth]{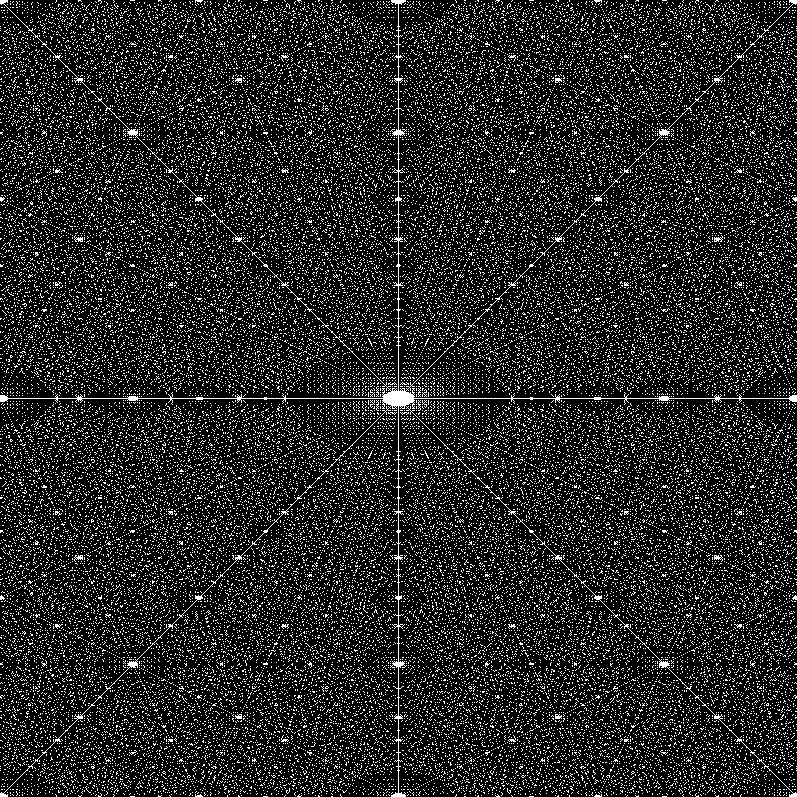}
  \caption{}
\end{subfigure}\par\medskip
\begin{subfigure}{\columnwidth}
  \centering
  \includegraphics[height=0.975\linewidth, angle=90]{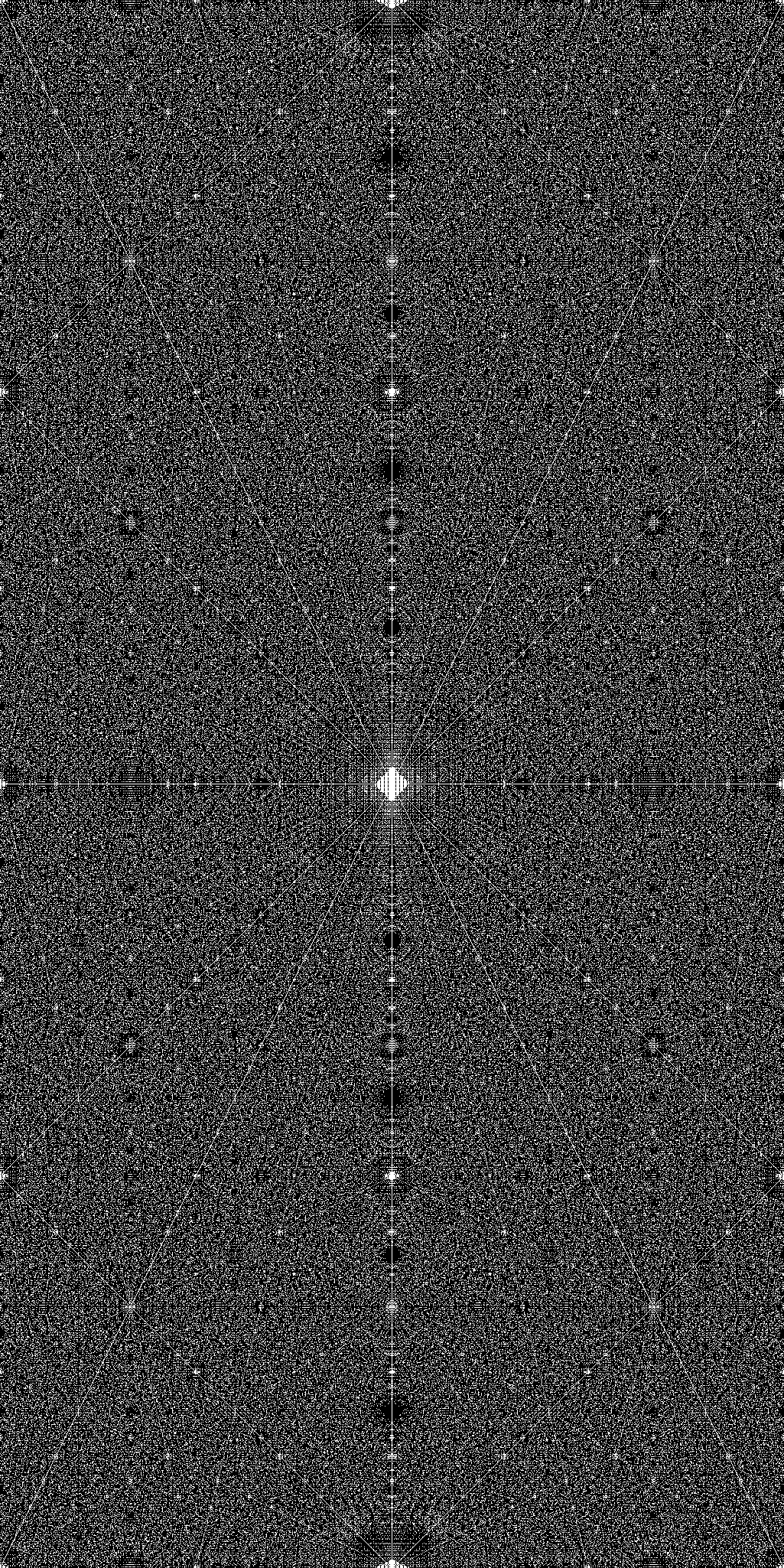}
  \caption{}
\end{subfigure}\par\medskip
\begin{subfigure}{0.5\columnwidth}
  \centering
  \includegraphics[width=0.95\linewidth]{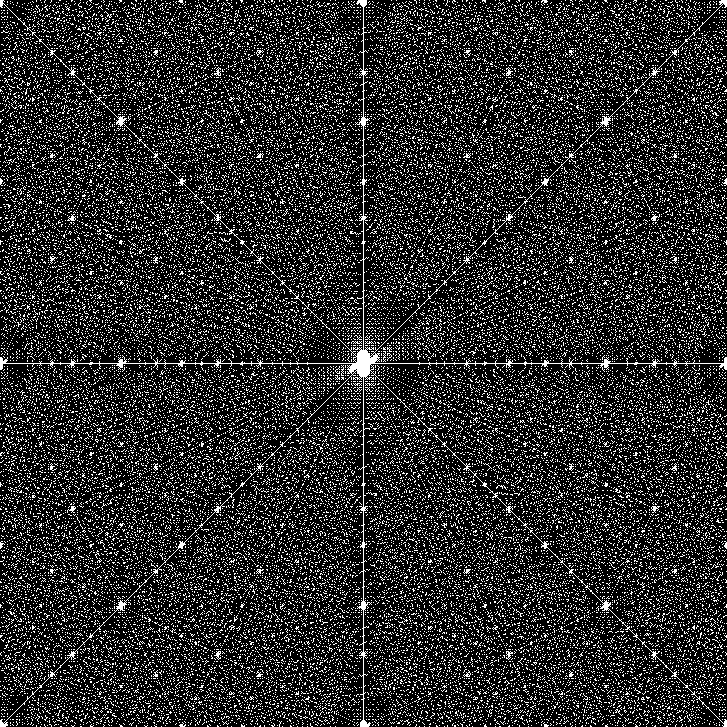}
  \caption{}
\end{subfigure}%
\begin{subfigure}{0.5\columnwidth}
  \centering
  \includegraphics[width=0.95\linewidth]{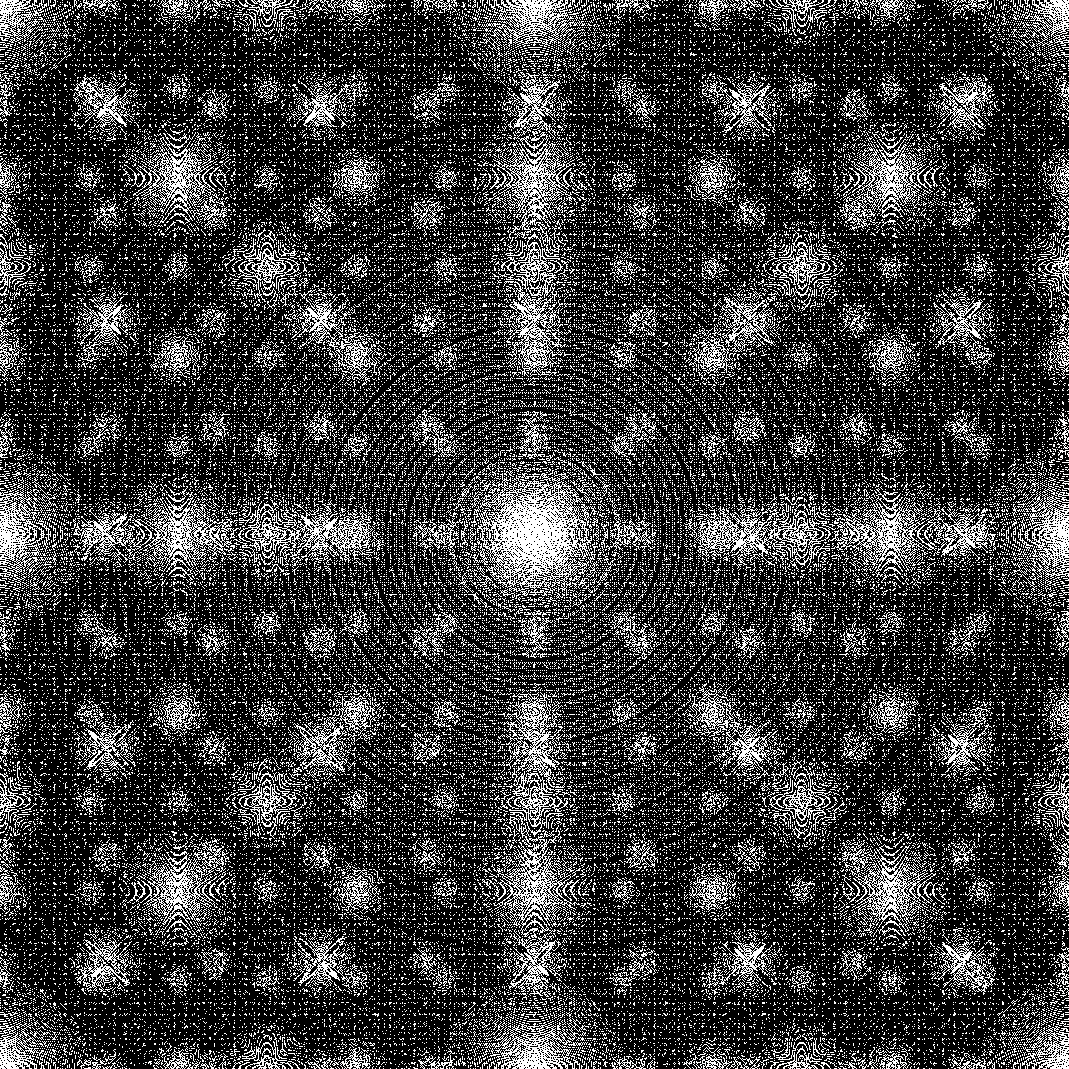}
  \caption{}
\end{subfigure}
\caption{A selection of novel \ac{ChaoS} fractals. These are generated: (a) using the $L^{0.5}$ norm after anti-clockwise rotation by $15^\circ$ on a $727\times727$ grid ($R = 145$); (b) using the $L^2$ norm after horizontal dilation by a factor of two on a $727\times727$ grid ($R = 121$); (c) using the $L^1$ norm on a $1025\times2049$ grid ($R=265)$; (d) using the polygon quasi-norm with an arbitrary star-shaped polygon on a $727\times727$ grid ($R=130$); and (e) explicitly from those multiples of a manually defined spiral pattern that give unity-smear \acp{KT} on a $1069\times1069$ grid.}
\label{fig: New fractals}
\end{figure}
\section{Conclusion}
In this paper the novel \acl{KT} was introduced, which formalises and extends upon the idea of downsampling and concatenating an image with itself. It was then demonstrated mathematically that certain scaling operations in the \ac{DFT} are equivalent to \acp{KT}. This was used to explain the previous fractal nature of the \ac{ChaoS} sampling pattern, and to present a variety of similar, novel patterns, along with methods for their construction for the \ac{DFT}. Demonstrating that such a rich family of fractal sampling patterns exists expands upon the potential utility of the \ac{DFT} for chaotic systems. In particular, it lays the foundation for a bespoke \ac{ChaoS} methodology, in which the sampling pattern used is tailored to the type of object being imaged, potentially improving speed and reconstruction.

\newpage
\printbibliography

\newpage
\section*{Supplementary Material}

\subsection{Motivating the Kaleidoscope Mapping}
\label{Supp A}

Previously, the definitions of the kaleidoscope mapping and transform were presented with minimal justification. Here, their motivation is presented in more detail.

Suppose we have a sequence of length $N$, $x[n]$, which we wish to downsampled by a factor of $\nu$ (for now, we assume that $\nu | N$). The canonical approach is to take every $\nu^\text{th}$ element of $x[n]$ beginning with $x[0]$ to form the new, downsampled subsequence; however, we could just as easily have started with $x[1]$, or $x[2]$, or so on, all the way up to $x[\nu - 1]$. This gives us $\nu$ possible downsampled subsequences, each with different elements. We denote such a downsampled subsequence starting at the $m^\text{th}$ element of $x[n]$ as $y_m[\ell]$. Each of these sequences has length $L = \frac{N}{\nu}$. We may then concatenate these possible downsampled subsequences together to give the sequence $z[k]$, also of length $N$. An example of this process is provided in Fig. \ref{fig: Kaleidoscope Flowchart}. 

\begin{figure}[ht]
\centering
\includegraphics[width = \linewidth]{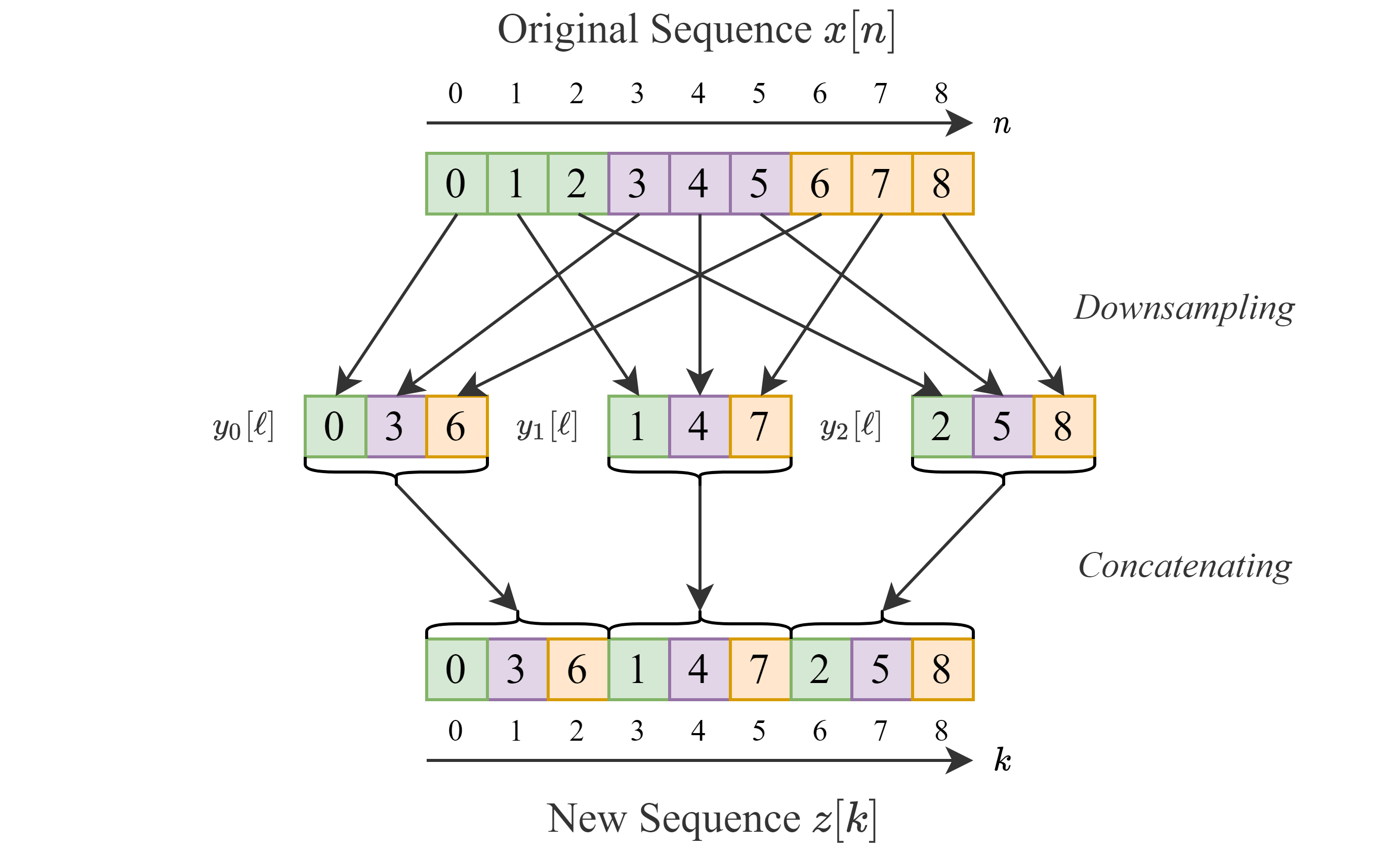}
\caption{Downsampling and concatenating $x[n] = n$ for $n \in \mathbb{Z}_9$ with downsampling factor $\nu=3$.}
\label{fig: Kaleidoscope Flowchart}
\end{figure}

The obvious question that then arises is how we may determine which positions in $x[n]$ correspond to which positions in $z[k]$. In other words, we seek a mapping $\kappa: n \mapsto k$. It is relatively straightforward to see that the final index $k$ corresponding to index $\ell$ of downsampled subsequence $y_m[\ell]$ is given by
\begin{align*}
    k &=Lm + \ell.
\end{align*}
Noting how each downsampled sequence is constructed, it then follows that
\begin{align*}
    k&=\frac{N}{\nu}(n \bmod \nu) + \floor{\frac{n}{\nu}}.
\end{align*}

This closely mirrors the form of the kaleidoscope mapping; however, there are some small differences. Notably, this mapping is limited to cases where $\nu | N$. There are two approaches we may take to generalising it to any positive integer downsampling factor, $\nu$. Firstly, we can round the length, $L$, of each subsequence up to the nearest integer by letting $L = \left \lceil \frac{N}{\nu} \right \rceil$. This would leave us with $\nu - 1$ subsequences of equal length, and one shorter subsequence. Alternatively, we can round $L$ down to the nearest integer by letting $L = \left \lfloor \frac{N}{\nu} \right \rfloor$. This would leave us with $\nu - 1$ subsequences of equal length and one longer subsequence. As stated in the body of the paper, it is useful to consider both of these approaches, and they correspond to the upper and lower kaleidoscope mappings respectively.

We may further generalise these mappings by introducing a smear factor, $\sigma$, which sets the spacing of each downsampled sequence when they are combined. When $\sigma = 1$, the sequences are simply concatenated, whereas other $\sigma$ values result in interleaving of the downsampled sequences, as demonstrated in Fig. \ref{fig: Smear factor example}. This generalisation is chiefly motivated by how it allows for the kaleidoscope-multiplication theorem to be stated and proved, as shown in the body of the paper.

\subsection{Additional Fractal Patterns}
\label{Supp B}

Moving away from the context of \ac{MRI}, the kaleidoscope transform may be used to generate both familiar and novel patterns, as presented in Fig. \ref{fig: Misc fractals}. Examples such as these demonstrate both the utility of the kaleidoscope transform in understanding existing self-similar patterns and its potential use in creating varied patterns for new tasks. 

\begin{figure}[ht]
\centering
\begin{subfigure}{0.5\columnwidth}
  \centering
  \includegraphics[width=0.95\linewidth]{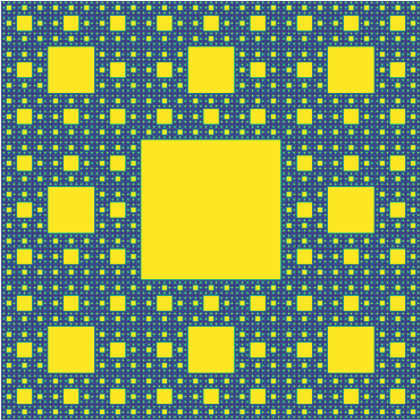}
  \caption{}
\end{subfigure}%
\begin{subfigure}{0.5\columnwidth}
  \centering
  \includegraphics[width=0.95\linewidth]{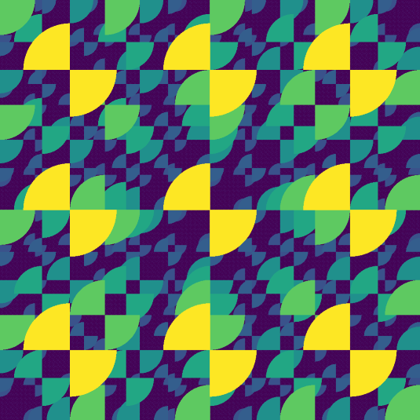}
  \caption{}
\end{subfigure}
\caption{Two patterns generated using kaleidoscope transforms: (a) an approximation to the Sierpinski carpet generated on a $728 \times 728$ grid, using a central pattern containing all of the vectors in the middle ninth of \ac{DFT} space and only kaleidoscope transforms with $\sigma = 1$; and (b) a bowtie-kaleidoscope pattern, made on a $1069 \times 1069$ grid, created using all of the vectors within a bowtie pattern one-sixth the width of \ac{DFT} space, and only kaleidoscope transforms with $\sigma = 1$, with each scaled image coloured differently.}
\label{fig: Misc fractals}
\end{figure}

\end{document}